\documentclass{article}
\usepackage{spconf,amsmath,graphicx}
\usepackage{amsfonts}
\usepackage{cite,graphicx,amsmath,amsthm}
\usepackage{subfigure}
\usepackage{fancyhdr}
\usepackage{dsfont}
\usepackage{array,color}
\usepackage{bm}
\usepackage{algorithm}
\usepackage{algpseudocode}

\newtheorem{lemma}{Lemma}

\newtheorem{proposition}{Proposition}

\newtheorem{corollary}{Corollary}

\newtheorem{property}{Property}

\newtheorem{remark}{Remark}

\newtheorem{claim}{Claim}

\title{Massive MIMO Multicast Beamforming\\Via Accelerated Random Coordinate Descent}
\name{Shuai Wang$^{\star}$ \qquad Lei Cheng$^{\star}$ \qquad Minghua Xia$^{\dagger}$\qquad Yik-Chung Wu$^{\star}$}

			\address{$^{\star}$ The University of Hong Kong, \{swang, leicheng, ycwu\}@eee.hku.hk \\
			    $^{\dagger}$ Sun Yat-sen University, xiamingh@mail.sysu.edu.cn}

\begin{document}

\maketitle
\begin{abstract}
One key feature of massive multiple-input multiple-output systems is the large number of antennas and users.
As a result, reducing the computational complexity of beamforming design becomes imperative.
To this end, the goal of this paper is to achieve a lower complexity order than that of existing beamforming methods, via the parallel accelerated random coordinate descent (ARCD).
However, it is known that ARCD is only applicable when the problem is convex, smooth, and separable.
In contrast, the beamforming design problem is nonconvex, nonsmooth, and nonseparable.
Despite these challenges, this paper shows that it is possible to incorporate ARCD for multicast beamforming by leveraging majorization minimization and strong duality.
Numerical results show that the proposed method reduces the execution time by one order of magnitude compared to state-of-the-art methods.

\end{abstract}


\section{Introduction}

Massive multiple-input multiple-output (MIMO) facilitates the concentration of wireless beams towards target directions \cite{infty}, and is a promising technology for 5G communication systems and beyond \cite{massive,massive2}.
On the other hand, in many emerging applications such as video streaming \cite{video} and computation offloading \cite{mec}, a large number of users could be interested in the same data, making massive MIMO multicast beamforming indispensible.

In the context of massive MIMO multicasting, a fundamental criterion for beamforming optimization is to minimize the power consumption subject to quality-of-service (QoS) constraints \cite{sdr,sdr2,review1,review2}.
However, due to high dimensionality (e.g., the number of antennas $N$ and the number of users $K$ can be in the range of hundreds or more \cite{massive,massive2}), traditional semidefinite relaxation (SDR) becomes extremely time-consuming, since the complexity of SDR is at least $O(KN^{3.5})$ \cite{sdr,sdr2,wang2}.
To reduce the computational complexity of beamforming optimization in large-scale settings, first-order methods (FOMs), which only involve the computation of gradients \cite{fom}, are recently proposed. In particular, the alternating direction method of multipliers (ADMM) has been derived in \cite{admm1,admm2} by introducing slack variables and solving the augmented Lagrangian problem. Nonetheless, the per-iteration complexity of ADMM is still $O(KN)$, and a fundamental question is: can we achieve a lower per-iteration complexity for the multicast beamforming optimization problem?

It turns out that this is possible if we only update one coordinate in each iteration \cite{cd}.
This leads to the coordinate descent method, which involves a per-iteration complexity of $O(K)$.
But unfortunately, the naive way of cyclicly updating the coordinates may diverge \cite{bcd}.
Even if it converges, the convergence rate can be slow \cite{cd}, thus offsetting the benefit brought by the low per-iteration complexity.
In fact, this is the reason why coordinate descent method received less attention compared to its full-gradient counterpart.

However, there has been a revival of interest in coordinate descent method recently \cite{rcd1,rcd2,rcd3}.
In particular, it is proved in \cite{rcd1} that if we randomly update one coordinate in each iteration, the resultant random coordinate descent (RCD) is guaranteed to converge with a rate of $O(1/m)$, where $m$ is the iteration counter.
Moreover, by adopting coordinate-wise momentum updates, RCD can be further accelerated to a faster convergence rate of $O\left(1/m^2\right)$ \cite{arcd1}.
This results in the parallel accelerated random coordinate descent (ARCD) method, which reduces the computation time by orders of magnitude compared to traditional FOMs in extensive applications, e.g., inverse problem \cite{arcd1}, $l_1$-regularized least squares problem \cite{arcd1}, supervised learning \cite{arcd2}, etc.

Nonetheless, ARCD can only be used to solve convex and smooth problems with separable constraints \cite{rcd1,rcd2,rcd3,arcd1,arcd2}.
In contrast, the multicast beamforming optimization problem is nonconvex with nonseparable constraints.
Therefore, we cannot directly apply ARCD to this application.
To this end, leveraging majorization minimization \cite{mm}, the nonconvex problem is transformed into a sequence of convex problems but with nonseparable constraints.
Furthermore, to resolve the coupling among different coordinates, the nonseparable problem is equivalently transformed into its Lagrangian dual counterpart, which is proved to be a coordinate-wise Lipschitz smooth problem with separable constraints, thus allowing ARCD to work on this dual problem.
Numerical results are further presented to demonstrate the low complexity nature of the proposed method.

\emph{Notation}.
Italic letter $x$, small bold letter $\mathbf{x}$, and capital bold letter $\mathbf{X}$ represent scalar, vector, and matrix, respectively.
The operator $[x]^+=\mathrm{max}(x,0)$, $\mathrm{Re}(x)$ takes the real part of $x$, and $|x|$ takes the modulus of $x$.
The symbol $[\mathbf{x}]_{k}$ takes the $k^{\mathrm{th}}$ element of vector $\mathbf{x}$, $[\mathbf{A}]_{k,:}$ is a row vector taking the $k^{\mathrm{th}}$ row of matrix $\mathbf{A}$, and $[\mathbf{A}]_{k,k}$ takes the diagonal element at the $k^{\mathrm{th}}$ row and the $k^{\mathrm{th}}$ column of matrix $\mathbf{A}$.
Finally, $\mathbb{E}(\cdot)$ represents the expectation of a random variable and $O(\cdot)$ represents the order of arithmetic operations.

\section{Problem Formulation and Existing Methods}

\setcounter{secnumdepth}{4}We consider a massive MIMO system consisting of a base station with $N$ antennas, and $K$ single-antenna users.
In particular, the base station transmits a signal $s$ with $\mathbb{E}[|s|^2]=1$ to all the users through the beamforming vector $\mathbf{v}\in\mathbb{C}^{N\times 1}$ with power $||\mathbf{v}||^2$.
Accordingly, the received signal $r_k\in \mathbb{C}$ at the user $k$ is
$r_{k}=\mathbf{h}^H_{k}\mathbf{v}s+n_{k}$, where $\mathbf{h}^H_{k}\in \mathbb{C}^{1\times N}$ is the downlink channel vector from the base station to user $k$, and $n_{k}\in \mathbb{C}$ is the Gaussian noise at the $k^{\mathrm{th}}$ user with power $\sigma_k^2$.
Based on the expression of $r_{k}$, the received SNR at user $k$ is $|\mathbf{g}^H_{k}\mathbf{v}|^2$, where $\mathbf{g}_{k}:=\mathbf{h}_{k}/\sigma_k$.

In multicast systems, our aim is to provide guaranteed SNR for all the users, while minimizing the total transmit power at base station: \cite{sdr}:
\begin{align}
&\mathrm{P}:\mathop{\mathrm{min}}_{\substack{\mathbf{v}}}
~||\mathbf{v}||_2^2~~~~\mathrm{s. t.}~~|\mathbf{g}^H_{k}\mathbf{v}|^2\geq \gamma,~~\forall k=1,\cdots,K,
\end{align}
where $\gamma$ is the common SNR target.
Problem $\mathrm{P}$ has been proved to be NP-hard in general \cite[Claim 1]{sdr}.
To solve $\mathrm{P}$, a traditional way is to apply SDR for convexification \cite{sdr}.
However, since SDR needs to solve a semidefinite programming (SDP) problem with $K$ variables and one semidefinite constraint of dimension $N\times N$, SDR requires a complexity of $O\left(\sqrt{N}(K^3+K^2N^2+KN^{3})\right)$ \cite{opt1}, which is too demanding when $N$ or $K$ is large.

To reduce the computational complexity of beamforming optimization, the majorization minimization (MM) framework \cite{sla,mm,admm2,wang3} can be adopted to transform $\mathrm{P}$ into a sequence of surrogate problems, 
Then, an iterative algorithm can be obtained with the following update at the $n^{\mathrm{th}}$ iteration:
\begin{align}
\mathrm{P}[n]:\mathbf{v}^{[n+1]}=\mathop{\mathrm{argmin}}_{\substack{\mathbf{v}}}~&\Bigg\{||\mathbf{v}||_2^2:
2\mathrm{Re}\left[\left(\mathbf{v}^{[n]}\right)^H\mathbf{g}_{k}\mathbf{g}^H_{k}\mathbf{v}\right]
\nonumber\\
&
-|\mathbf{g}^H_{k}\mathbf{v}^{[n]}|^2\geq \gamma,~~\forall k
\Bigg\}.
\end{align}
It has been proved in \cite{sla} that the sequence $\{\mathbf{v}^{[0]},\mathbf{v}^{[1]},\cdots\}$ converges to a Kruash-Kuhn-Tucker solution to $\mathrm{P}$.

However, even capitalizing on the MM framework, $\mathrm{P}[n]$ is still large-scale, and the interior point method (IPM) adopted in \cite{sla} would lead to time-consuming computations since the complexity of IPM is $O\left(\sqrt{K}(N^{3}+2NK)\right)$ \cite{opt1}.
To overcome this challenge, ADMM has been proposed for solving $\mathrm{P}[n]$ \cite{admm1,admm2}.
Such a method reformulates $\mathrm{P}[n]$ into an augmented Lagrangian problem and then uses hybrid gradient method to solve it.
Therefore, its per-iteration complexity is only $O(KN)$.

\section{ARCD: Accelerated Random Coordinate Descent}

While the per-iteration complexity of ADMM is lower than that of SDR, a natural question is: can we achieve a lower complexity than that of ADMM for solving $\mathrm{P}[n]$?
This section will show that it is possible under the framework of ARCD.
To begin with, the following property is established.
\begin{property}
Strong duality holds for $\mathrm{P}[n]$.
\end{property}
\begin{proof}
To prove the strong duality of $\mathrm{P}[n]$, it suffices to show that $\mathrm{P}[n]$ is convex and satisfies Slater's condition \cite{opt3}.
Since the objective  $||\mathbf{v}||_2^2$ of $\mathrm{P}[n]$ is convex quadratic and the constraints of $\mathrm{P}[n]$ are linear, $\mathrm{P}[n]$ is convex.
On the other hand, showing that $\mathrm{P}[n]$ satisfies Slater's condition is equivalent to finding a feasible point of $\mathrm{P}[n]$ satisfying all the constraints with strict inequality.
To this end, consider the solution $\mathbf{v}=(1+\chi)\cdot\mathbf{v}^{[n]}$ with $\chi>0$, and it can be shown that $2\mathrm{Re}\left[\left(\mathbf{v}^{[n]}\right)^H\mathbf{g}_{k}\mathbf{g}^H_{k}\mathbf{v}'\right]
-|\mathbf{g}^H_{k}\mathbf{v}^{[n]}|^2>\gamma$ for all $k$.
This completes the proof.
\end{proof}
Based on the result of \textbf{Property 1}, the dual problem of $\mathrm{P}[n]$ must have the same optimal value as $\mathrm{P}[n]$ \cite{opt3}.
Therefore, we propose to transform $\mathrm{P}[n]$ into its Lagrangian dual domain, which gives the following proposition.
\begin{proposition}
The dual problem of $\mathrm{P}[n]$ is
\begin{align}
&\mathrm{D}:\mathop{\mathrm{max}}_{\substack{\mathbf{q}\succeq \mathbf{0}}}
~-\Upsilon^{[n]}\left(\mathbf{q}\right),
\end{align}
where $\mathbf{q}=[q_1,\cdots,q_K]^T\in\mathbb{R}^{K\times 1}$ and
\begin{align}\label{Upsilon}
\Upsilon^{[n]}\left(\mathbf{q}\right)=&
\Big|\Big|\sum_{k=1}^Kq_k\cdot\mathbf{g}_{k}\mathbf{g}^H_{k}\mathbf{v}^{[n]}\Big|\Big|_2^2
\nonumber\\
&-\sum_{k=1}^Kq_{k}\left(\gamma+|\mathbf{g}_{k}^H\mathbf{v}^{[n]}|^2\right).
\end{align}
Moreover, denoting the optimal solution of $\mathbf{q}$ to $\mathrm{D}$ as $\mathbf{q}^*$, the optimal $\mathbf{v}^*$ of $\mathrm{P}[n]$ is
\begin{align}
&\mathbf{v}^*=\sum_{k=1}^Kq^{*}_k\cdot\mathbf{g}_{k}\mathbf{g}^H_{k}\mathbf{v}^{[n]}. \label{v*3}
\end{align}
\end{proposition}
\begin{proof}
The Lagrangian of $\mathrm{P}[n]$ is
\begin{align}
&\mathcal{L}(\mathbf{v},\mathbf{q})=
||\mathbf{v}||_2^2
\nonumber\\
&
+\mathop{\sum}_{k=1}^Kq_k\left\{\gamma-2\mathrm{Re}\left[\left(\mathbf{v}^{[n]}\right)^H\mathbf{g}_{k}\mathbf{g}^H_{k}\mathbf{v}\right]
+|\mathbf{g}^H_{k}\mathbf{v}^{[n]}|^2
\right\}, \nonumber
\end{align}
With the Lagrangian $\mathcal{L}$, the dual problem of $\mathrm{P}[n]$ is given by \cite{opt3}
\begin{align}
&\mathop{\mathrm{max}}_{\mathbf{q}\succeq \mathbf{0}}~\mathop{\mathrm{min}}_{\mathbf{v}}~\mathcal{L}\left(\mathbf{v},\mathbf{q}\right). \label{B1}
\end{align}
To compute $\mathop{\mathrm{min}}_{\mathbf{v}}~\mathcal{L}\left(\mathbf{v},\mathbf{q}\right)$, we set $\partial \mathcal{L}/\partial~\mathrm{conj}(\mathbf{v})=0$, and obtain
$\mathbf{v}^*=\sum_{k=1}^Kq_k\cdot\mathbf{g}_{k}\mathbf{g}^H_{k}\mathbf{v}^{[n]}$, which gives \eqref{v*3}.
Finally, putting $\mathbf{v}^*$ in \eqref{v*3} into \eqref{B1}, the objective function of \eqref{B1} is
\begin{align}
\mathop{\mathrm{min}}_{\mathbf{v}}~\mathcal{L}\left(\mathbf{v},\mathbf{q}\right)&=
\mathcal{L}\left(\sum_{k=1}^Kq_k\cdot\mathbf{g}_{k}\mathbf{g}^H_{k}\mathbf{v}^{[n]},\mathbf{q}\right)
=
-\Upsilon^{[n]}\left(\mathbf{q}\right), \nonumber
\end{align}
and the problem \eqref{B1} is equivalently written as $\mathrm{D}$.
\end{proof}

Based on \textbf{Proposition 1} and by defining
\begin{align}
&\mathbf{F}=\left[\mathbf{g}_{1}\mathbf{g}^H_{1}\mathbf{v}^{[n]},\cdots,\mathbf{g}_{K}\mathbf{g}^H_{K}\mathbf{v}^{[n]}\right]\in\mathbb{C}^{N\times K}, \label{F}
\\
&\mathbf{d}=\left[\gamma+|\mathbf{g}^H_{1}\mathbf{v}^{[n]}|^2,\cdots,\gamma+|\mathbf{g}^H_{K}\mathbf{v}^{[n]}|^2\right]^T
\in\mathbb{C}^{K\times 1}, \label{c}
\end{align}
the function $\Upsilon^{[n]}\left(\mathbf{q}\right)$ in $\mathrm{D}$ can be re-written as
$\Upsilon^{[n]}\left(\mathbf{q}\right)=||\mathbf{F}\mathbf{q}||^2-\mathbf{d}^T\mathbf{q}$, which is quadratic.
Therefore, FOMs (e.g., gradient descent, accelerated gradient projection, etc.) can be adopted by computing $\nabla\Upsilon^{[n]}\left(\mathbf{q}\right)$ in each iteration, which require a complexity of $O(KN)$.
However, in the following, a stronger property of $\Upsilon^{[n]}(\mathbf{q})$ will be established, thus allowing more efficient updates with per-iteration complexity smaller than $O(KN)$.
\begin{property}
The function $\Upsilon^{[n]}(\mathbf{q})$ in $\mathrm{D}$ is coordinate-wise $L_k$-smooth, where
\begin{align}
L_{k}&=\left[2\mathrm{Re}\left(\mathbf{F}^H\mathbf{F}\right)\right]_{k,k}.
\label{Lip3}
\end{align}
\end{property}
\begin{proof}
According to \cite{rcd1,rcd2}, $\Upsilon^{[n]}(\mathbf{q})$ is coordinate-wise $L_k$-smooth if and only if
\begin{align}\label{cwsmooth}
&\Big|\left[\nabla_{\mathbf{q}} \Upsilon^{[n]}(\mathbf{q}+t\cdot\mathbf{e}_k)\right]_k-\left[\nabla_{\mathbf{q}} \Upsilon^{[n]}(\mathbf{q})\right]_k \Big|\leq L_k|t|,
\end{align}
where $t\in\mathbb{R}$ and $\mathbf{e}_{k}$ represents the unit vector with the $k^{\mathrm{th}}$ element being 1 and others being zero.
By computing
\begin{align}\label{gradient3}
&\nabla_{\mathbf{q}}\Upsilon^{[n]}\left(\mathbf{q}\right)=
2\mathrm{Re}\left(\mathbf{F}^H\mathbf{F}\mathbf{q}\right)-\mathbf{d},
\end{align}
and putting $\nabla_{\mathbf{q}} \Upsilon^{[n]}(\mathbf{q}+t\cdot\mathbf{e}_k)$ and $\nabla_{\mathbf{q}} \Upsilon^{[n]}(\mathbf{q})$ into \eqref{cwsmooth}, 
the left hand side of \eqref{cwsmooth} becomes $\left[2\mathrm{Re}\left(\mathbf{F}^H\mathbf{F}\right)\right]_{k,k}|t|$ and the property is immediately proved.
\end{proof}

Based on \textbf{Property 2} and since the constraint $\mathbf{q}\succeq \mathbf{0}$ is separable (it can be re-written as $q_k\geq 0$ for all $k$), problem $\mathrm{D}$ can be solved by the coordinate descent method.
In particular, at iteration $m$, we first choose a coordinate, say $l$, and then update
\begin{align}
&\left[\mathbf{q}^{[m+1]}\right]_l=\left(\left[\mathbf{q}^{[m]}\right]_l-\frac{1}{L_l}
\left[\nabla_{\mathbf{q}} \Upsilon^{[n]}(\mathbf{q}^{[m]})\right]_l\right)^+,\label{rcd}
\end{align}
where the $l^{\mathrm{th}}$ coordinate-wise gradient
\begin{align}\label{cwg}
\left[\nabla_{\mathbf{q}} \Upsilon^{[n]}(\mathbf{q}^{[m]})\right]_l=&
2\mathrm{Re}\left[\left(\mathbf{F}^H\mathbf{F}\right)_{l,:}\mathbf{q}^{[m]}\right]-[\mathbf{d}]_l.
\end{align}
It can be seen from \eqref{rcd} and \eqref{cwg} that the complexity of updating \eqref{rcd} is only $O(K)$.

The next question is how to choose the index $l$.
A straightforward idea is to choose $l$ in a cyclic manner, i.e., $l=1,2,\cdots,K,1,2,\cdots,K$.
But this may lead to slow convergence of the algorithm \cite{cd}.
A better way is to choose $l$ as a random integer number ranging from $1$ to $K$ with equal probability, i.e., $\mathrm{Pr}(l=1)=\cdots=\mathrm{Pr}(l=K)=1/K$.
With such a random schedule and the coordinate update in \eqref{rcd}, the sequence $\{\mathbf{q}^{[0]},\mathbf{q}^{[1]},\cdots\}$ converges to the optimal $\mathbf{q}^*$ to $\mathrm{D}$ with a convergence rate of $O(1/m)$ \cite{rcd1,rcd2}.
This is the so-called RCD method.

However, the RCD method can still be improved in the following two ways.
First, the update of \eqref{rcd} is sequential, which means that the next iteration must wait before the current iteration is completed.
If a block of coordinates is updated in parallel, the running time of RCD can be further reduced \cite{rcd2,rcd3}.
Second, there is a gap between the convergence rate $O(1/m)$ of RCD and the best known convergence rate $O(1/m^2)$ \cite{yurii1} for solving smooth problems.
This indicates that we can consider adding momentums to accelerate the convergence of RCD \cite{arcd1}.

Based on the above observations, the following ARCD is adopted for solving $\mathrm{D}$.
More specifically, at the $m^{\mathrm{th}}$ iteration, instead of generating a random number $l$, we generate a random set $\mathcal{Y}^{[m]}$ with
\begin{align}
&|\mathcal{Y}^{[m]}|=Y,~~~~~\mathcal{Y}^{[m]}\subseteq\{1,\cdots,K\},
\nonumber\\
&\mathrm{Pr}(1\in\mathcal{Y}^{[m]})=\cdots=\mathrm{Pr}(K\in\mathcal{Y}^{[m]}), \label{Ym}
\end{align}
and update all the coordinates in $\mathcal{Y}^{[m]}$ as
\begin{align}
&\left[\mathbf{q}^{[m+1]}\right]_i=\Bigg\{\left[\mathbf{q}^{[m]}\right]_i-\frac{1}{Kc^{[m]}L_i}
\nonumber\\
&
\times\left[\nabla_{\mathbf{q}} \Upsilon^{[n]}\left(\mathbf{q}^{[m]}+(c^{[m]})^2\mathbf{z}^{[m]}\right)\right]_i\Bigg\}^+,
~~\forall i\in \mathcal{Y}_m,\label{arcd}
\end{align}
where the momentum $(c^{[m]})^2\mathbf{z}^{[m]}$ is added, with $\mathbf{z}^{[m]}$ being the direction at the $m^{\mathrm{th}}$ iteration and $c^{[m]}$ being the step-size to control the importance of $\mathbf{z}^{[m]}$.

It can be seen from \eqref{arcd} that different coordinates in $\mathcal{Y}^{[m]}$ are updated in parallel.
On the other hand, by choosing \begin{align}
c^{[m+1]}=&\frac{1}{2}\left(\sqrt{(c^{[m]})^4+4(c^{[m]})^2}-(c^{[m]})^2\right), \label{cm}
\\
\left[\mathbf{z}^{[m+1]}\right]_i=&\left[\mathbf{z}^{[m]}\right]_i-\frac{1}{(c^{[m]})^2}\cdot\left(1-\dfrac{Kc^{[m]}}{Y}\right)
\nonumber\\
&\times
\left(\left[\mathbf{q}^{[m+1]}\right]_i-\left[\mathbf{q}^{[m]}\right]_i\right),~~\forall i\in \mathcal{Y}_m, \label{z}
\end{align}
with the initial $c^{[0]}=Y/K$ and $\mathbf{z}^{[0]}=\mathbf{0}$, $\mathbf{q}^{[m+1]}$ computed using \eqref{arcd} is guaranteed to converge to the optimal solution to $\mathrm{D}$ with a convergence rate $O(1/m^2)$ \cite[Theorem 3]{arcd1}.
Since $\mathrm{D}$ is equivalent to $\mathrm{P}[n]$ and $\mathrm{P}[n]$ represents a surrogate problem for $\mathrm{P}$, problem $\mathrm{P}$ can be solved via MM and ARCD in the Lagrangian dual domain. The entire procedure is summarized in Algorithm 1.

\begin{algorithm}[!t]
    \caption{Solving $\mathrm{P}$ via ARCD.}
        \begin{algorithmic}[1]
 \State \textbf{Input} $\{\mathbf{g}_k,\gamma\}$.
            \State  Initialize $\mathbf{v}^{[0]}$ with a feasible $\mathbf{v}$. Set $\mathbf{q}^{*}=\mathbf{0}$ and $n=0$.
            \State \textbf{Repeat} (MM iteration)
            \State \ \ \ Compute $\mathbf{F}$ in \eqref{F} and $\mathbf{d}$ in \eqref{c}.
            \State \ \ \ Initialize $\mathbf{q}^{[0]}=\mathbf{q}^{*}$ and set $Y=K/5$.
            \State \ \ \ Set $c^{[0]}=Y/K$ and $\mathbf{z}^{[0]}=\mathbf{0}$. Set $m=0$.
            \State \ \ \ \textbf{Repeat} (ARCD iteration)
            \State \ \ \ \ \ \ \textbf{Generate} a random set $\mathcal{Y}_m$ according to \eqref{Ym}.
            \State \ \ \ \ \ \ \textbf{Parallel For} $i\in\mathcal{Y}_m$
            \State \ \ \ \ \ \ \ \ \ Update $\left[\mathbf{q}^{[m+1]}\right]_i$ according to \eqref{arcd}.
            \State \ \ \ \ \ \ \ \ \ Update $\left[\mathbf{z}^{[m+1]}\right]_i$ according to \eqref{z}.
            \State \ \ \ \ \ \ \textbf{End}
            \State \ \ \ \ \ \ Update $c^{[m+1]}$ according to \eqref{cm}.
            \State \ \ \ \ \ \ Set $m:=m+1$.
            \State \ \ \ \textbf{Until} $|\Upsilon^{[n]}(\mathbf{q}^{[m]})-\Upsilon^{[n]}(\mathbf{q}^{[m-1]})|<10^{-7}$.
            \State \ \ \ Set $\mathbf{q}^{*}=\mathbf{q}^{[m]}$.
            \State \ \ \ Set $\mathbf{v}^{[n+1]}=\sum_{k=1}^Kq_k^{*}\mathbf{g}_{k}\mathbf{g}^H_{k}\mathbf{v}^{[n]}$.
            \State \ \ \ Set $n:=n+1$.
            \State \textbf{Until} the stopping criterion of MM is met.
        \end{algorithmic}
\end{algorithm}

\section{Simulation Results}

This section presents simulation results to verify the performance of the proposed ARCD.
In particular, each random channel is generated according to $\mathcal{CN}(\mathbf{0},\varrho\mathbf{I})$ \cite{wang3},
where the path-loss is $\varrho=-90~\mathrm{dB}$.
It is assumed that the noise power $\sigma^2_1=\cdots=\sigma^2_K=-80~\mathrm{dBm}$, which includes thermal noise and receiver noise \cite{wireless}.
Each point in the figures is obtained by averaging over $100$ simulation runs, with independent channels between consecutive runs.
All problem instances are solved by Matlab R2015b on a desktop with Intel Core i5-4570 CPU at 3.2 GHz and 8GB RAM.
For comparison, we also simulate the MM-IPM\footnote{The MM-IPM \cite{sla} is implemented using the Matlab software CVX Mosek \cite{opt3}.} \cite{sla}, the ADMM \footnote{The ADMM is implemented by introducing slack variables $\{\mathbf{w}_k=\mathbf{v}\}_{k=1}^K$ \cite{admm1} and adding a consensus penalty $a/2\sum_{k=1}^K||\mathbf{v}-\mathbf{w}_k+\bm{\eta}_k||_2^2$, with $a=2/\sqrt{N}$ \cite{admm2} and the dual variables being $\{\bm{\eta}_k\}$.} \cite{admm1,admm2}, and the asymptotic method\footnote{The asymptotic solution is implemented by assuming $N\rightarrow+\infty$ \cite{infty}, and this method is used as initialization for the other simulated methods.} \cite{infty}.

To evaluate the solution quality and the running time of Algorithm 1, we simulate the case of $N=200$ with $K\in\{50, 100, 200, 500\}$. Notice that the case of large $K$ is very important for future crowd sensing applications.
For the methods based on MM, the number of MM iterations is $20$ \cite{sla,admm1,admm2}.
Furthermore, the ADMM for solving $\mathrm{P}[n]$ stops when the change of objective functions between consecutive iterations is smaller than $\tau=10^{-5}$ \cite{admm1}.
If ADMM fails to reach the above condition within $2000$ iterations, it stops and outputs the result at iteration $2000$ \cite{admm1,admm2}.
It can be observed from Fig. 1a that Algorithm 1 significantly outperforms the asymptotic solution \cite{infty}, and slightly outperforms the ADMM \cite{admm1,admm2}.
In fact, Algorithm 1 achieves the same power consumption as the MM-IPM \cite{sla}.
However, as illustrated in Fig. 1b, Algorithm 1 only requires less than $2$ seconds to finish for all the simulated value of $K$.
Compared to MM-IPM and ADMM, Algorithm 1 saves at least $90\%$ (one order of magnitude) of the computation times.

\begin{figure}
  \centering
  \subfigure[]{
    \label{fig:subfig:a} 
    \includegraphics[width=35mm]{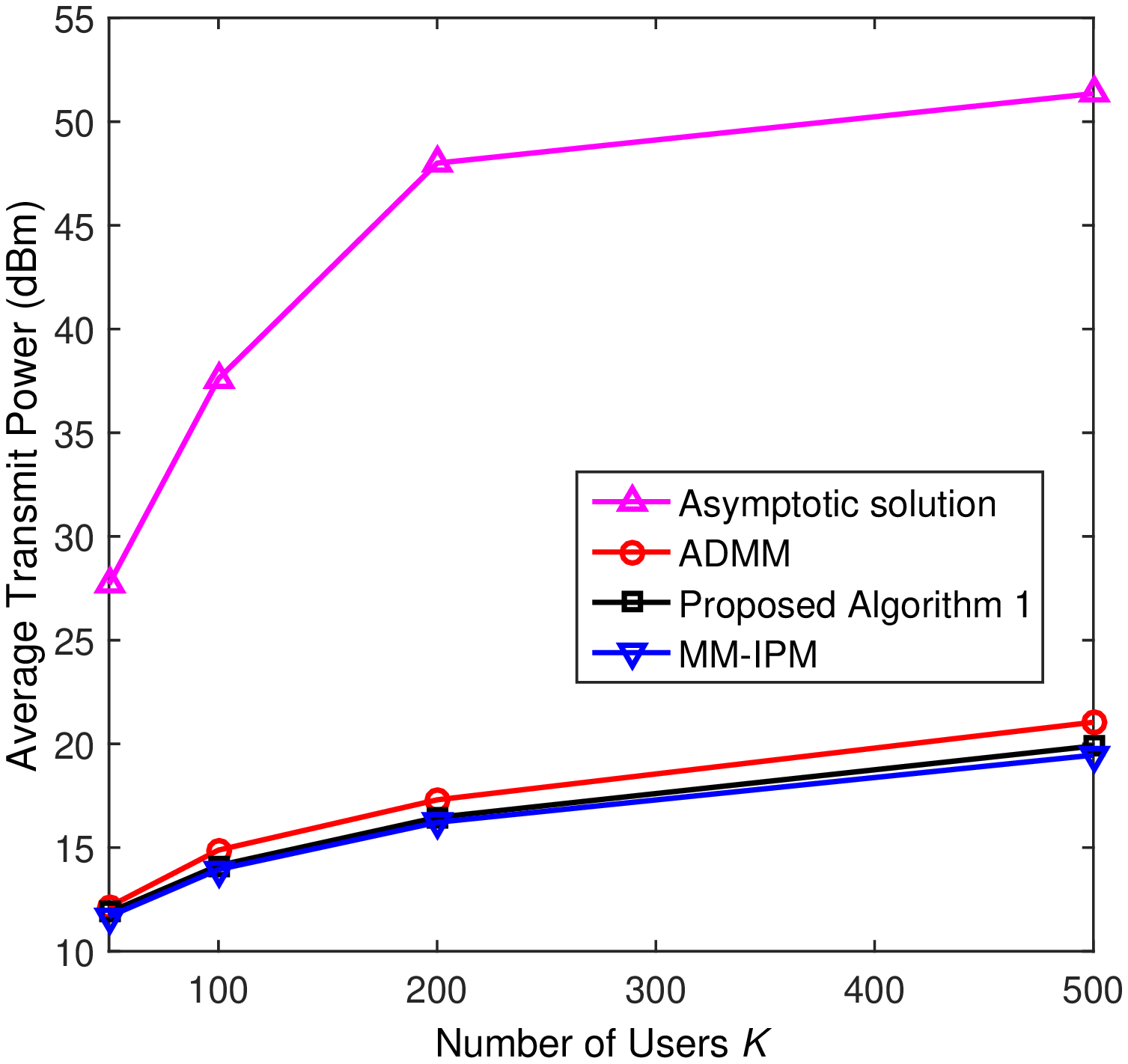}}
  \hspace{0in}
  \subfigure[]{
    \label{fig:subfig:b} 
    \includegraphics[width=35mm]{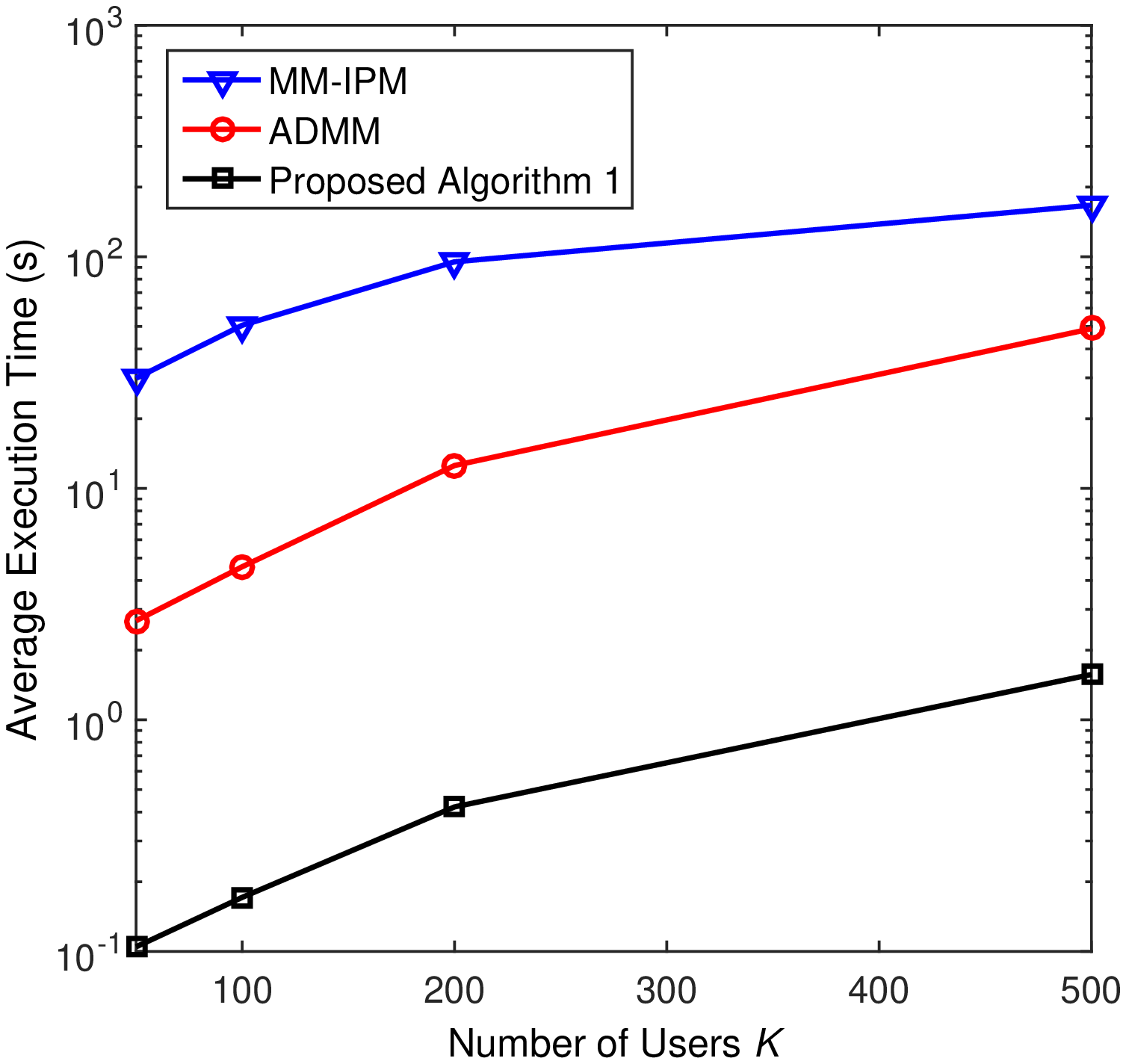}}
  \caption{(a) Transmit power in dBm versus number of users $K$ for the case of $N=200$ and $\gamma=10~\mathrm{dB}$; (b) Average execution time versus number of users $K$ for the case of $N=200$ and $\gamma=10~\mathrm{dB}$.}
  \label{fig:subfig} 
\end{figure}

\section{Conclusions}

This paper studied the massive MIMO multicast beamforming optimization problem.
With majorization minimization and strong duality, the primal problem was transformed into a coordinate-wise Lipschitz smooth problem with separable constraints.
By further adopting ARCD, lower complexity than that of existing algorithms was achieved.
Simulation results showed that the proposed method reduces the execution time by one order of magnitude compared to existing methods while guaranteeing the same performance.

\bibliographystyle{IEEEbib}
\bibliography{strings,refs}

\end{document}